\documentclass[final,11pt]{article}
\usepackage[utf8]{inputenc}
\usepackage{cite}
\usepackage{fullpage}
\usepackage{amssymb}
\usepackage{amsmath,amsthm}
\usepackage{graphicx}
\usepackage{xspace}
\usepackage{hyperref}
\usepackage{subfigure}
\usepackage{ifdraft}

\usepackage{enumerate}

\usepackage[noend]{algorithmic}
\usepackage{algorithm}

\usepackage[inline,nomargin]{fixme}
\fxsetface{inline}{\bfseries\color{red}}

\usepackage{tikz}
\usetikzlibrary{calc}
\usetikzlibrary{shapes.misc}
\usetikzlibrary{decorations}
\usetikzlibrary{decorations.markings}

\tikzstyle{defnode}=[circle,fill=white,draw=black,inner sep=0pt,minimum size=5pt]
\tikzstyle{ionode}=[circle,fill=black,draw=black,inner sep=0pt,minimum size=5pt]
\tikzstyle{dedge}=[double,double distance=1.5pt]
\tikzstyle{tedge}=[double,double distance=2pt]

\newtheorem{theorem}{Theorem}
\newtheorem{lemma}{Lemma}
\newtheorem{corollary}{Corollary}

\newcommand{\fotwocolor}{FO-2-Color}
\title{The Hardness of the Functional Orientation 2-Color Problem}
{\author{\renewcommand{\thefootnote}{\arabic{footnote}} Søren Bøg\footnotemark[1]  ~~~~~~~~ Morten Stöckel\footnotemark[2]  ~~~~~~~~ Hjalte Wedel Vildhøj\footnotemark[1]\\[1em]
\renewcommand{\thefootnote}{\arabic{footnote}}\footnotemark[1]~~\small Technical University of Denmark, DTU Compute, \texttt{\{sbog@dtu.dk,hwv@hwv.dk\} }\\
\renewcommand{\thefootnote}{\arabic{footnote}}\footnotemark[2]~~\small IT University of Copenhagen, \texttt{mstc@itu.dk} } }
\date{}

\begin{document}
\renewcommand{\thefootnote}{\arabic{footnote}}
\maketitle
\ifdraft{\begin{center}\huge \textit{DRAFT}\\Compiled \today\end{center}}{}

\begin{abstract}
We consider the Functional Orientation 2-Color problem, which was introduced by Valiant in his seminal paper on holographic algorithms~[SIAM J. Comput., 37(5), 2008]. For this decision problem, Valiant gave a polynomial time holographic algorithm for planar graphs of maximum degree 3, and showed that the problem is \textbf{NP}-complete for planar graphs of maximum degree 10. A recent result on defective graph coloring by Corrêa~et~al. [Australas.~J.~Combin., 43, 2009] implies that the problem is already hard for planar graphs of maximum degree 8. Together, these results leave open the hardness question for graphs of maximum degree between 4 and 7.

We close this gap by showing that the answer is always yes for arbitrary graphs of maximum degree 5, and that the problem is \textbf{NP}-complete for planar graphs of maximum degree 6. Moreover, for graphs of maximum degree 5, we note that a linear time algorithm for finding a solution exists.
\end{abstract}

\section{Introduction}
In this paper we assume $G=(V,E)$ is an undirected multigraph without loops. The maximum degree of $G$ is denoted $\Delta$. An \emph{isolated vertex} is a vertex with degree zero. A \emph{functional orientation} of $G$ is an assignment of directions to a set of edges such that every non-isolated vertex has exactly one edge directed away from it. A single edge may be assigned two opposite directions, or it may remain undirected. A \emph{full functional orientation} is a functional orientation of $G$ that leaves no edges of $G$ undirected. A \emph{$k$-coloring} of $G$ is a partition of $V$ into $k$ sets $V_1,V_2,\ldots, V_k$. If the subgraphs induced by $V_1,V_2,\ldots,V_k$ contain no edges, we say that the coloring is a \emph{proper $k$-coloring}. For a given $k$-coloring, an \emph{induced monochromatic component} is a connected component of the subgraph induced by $V_i$ for some $i=1,\ldots,k$.

Functional orientations occur in many applications. They naturally capture deterministic transitional systems such as finite state machines and positional strategies in, e.g., stochastic games and Markov decision processes~\cite{shapley1953stochastic, andersson2009complexity}. Functional orientations also prove useful in the analysis of algorithms, for example, in a dictionary with cuckoo hashing a new element can be inserted without rehashing if and only if the resulting cuckoo graph has a full functional orientation~\cite{pagh2004cuckoo}. The more general concept of \emph{$c$-orientations}, where the edges are oriented such that every vertex has out-degree at most $c$, has previously been studied in connection with dynamic representations of sparse graphs~\cite{brodal1999dynamic}.

The \emph{Functional Orientation 2-Color problem} (\fotwocolor\ problem) is to determine if $G$ has a (not necessarily proper) 2-coloring of the vertices and a functional orientation, such that every edge between two vertices of the same color is directed in at least one direction by the functional orientation. Equivalently, a graph is \fotwocolor able if and only if there is a 2-coloring of $G$ such that every induced monochromatic component has a full functional orientation. See \autoref{fig:example} for an example.

\begin{figure}
\centering
\hfill
\subfigure[]{
\begin{tikzpicture}[scale=0.9,every node/.style=defnode]
\begin{scope}[decoration={
markings,
mark=at position 0.5 with {\arrow[black,line width=1]{stealth};}},
every edge/.append style={postaction={decorate}}]
\draw (-30:1) node(1)[fill=black] {};
\draw (90:1) node(2)[fill=black] {};
\draw (210:1) node(3) {};
\draw (-30:2) node(4) {};
\draw (90:2) node(5) {};
\draw (210:2) node(6)[fill=black] {};
\draw (1) edge [bend left=10] (2);
\draw (2) edge [bend left=10] (1);
\draw (2) to [bend left=10] (3);
\draw (2) to [bend right=10] (3);
\draw (3) to [bend left=10] (1);
\draw (3) to [bend right=10] (1);
\draw (4) edge [bend left=10] (5);
\draw (5) edge [bend left=10] (4);
\draw (5) to [bend left=10] (6);
\draw (5) to [bend right=10] (6);
\draw (6) to [bend left=10] (4);
\draw (6) to [bend right=10] (4);
\draw (1) to (4);
\draw (2) to (5);
\end{scope}

\draw (3) edge[decoration={
markings,
mark=at position 0.65 with {\arrow[black,line width=1]{stealth};},
mark=at position 0.35 with {\arrowreversed[black,line width=1]{stealth};}},postaction={decorate}] (6);
\end{tikzpicture}
\label{fig:introvalidexample}
}
\hfill
\subfigure[]{
\begin{tikzpicture}[scale=0.9,every node/.style=defnode]
\draw (-30:1) node(1)[fill=black] {};
\draw (90:1) node(2)[fill=black] {};
\draw (210:1) node(3)[fill=black] {};
\draw (-30:2) node(4) {};
\draw (90:2) node(5) {};
\draw (210:2) node(6) {};
\draw (1) to [bend left=10] (2);
\draw (1) to [bend right=10] (2);
\draw (2) to [bend left=10] (3);
\draw (2) to [bend right=10] (3);
\draw (3) to [bend left=10] (1);
\draw (3) to [bend right=10] (1);
\draw (4) edge [bend left=10] (5);
\draw (4) edge [bend right=10] (5);
\draw (5) to [bend left=10] (6);
\draw (5) to [bend right=10] (6);
\draw (6) to [bend left=10] (4);
\draw (6) to [bend right=10] (4);
\draw (1) to (4);
\draw (2) to (5);
\draw (3) to (6);
\end{tikzpicture}\label{fig:introinvalidexample}
}
\hfill
\subfigure[]{
\begin{tikzpicture}[scale=0.9,every node/.style=defnode]
\draw (-30:2) node(1) {};
\draw (90:2) node(2) {};
\draw (210:2) node(3) {};
\draw (1) to [bend left=10] (2);
\draw (1) to (2);
\draw (1) to [bend right=10] (2);
\draw (2) to [bend left=10] (3);
\draw (2) to (3);
\draw (2) to [bend right=10] (3);
\draw (3) to [bend left=10] (1);
\draw (3) to (1);
\draw (3) to [bend right=10] (1);
\end{tikzpicture}\label{fig:intronosolution}
}
\hfill
\caption{Two different 2-colorings of a graph with maximum degree 5 are shown in \subref{fig:introvalidexample} and \subref{fig:introinvalidexample}. The coloring in \subref{fig:introvalidexample} admits a full functional orientation, such as the one shown, on the induced subgraphs.  On the other hand, the subgraphs induced by the coloring in \subref{fig:introinvalidexample} have no full functional orientation. The graph in \subref{fig:intronosolution} is not \fotwocolor able.\label{fig:example}}
\end{figure}

The \fotwocolor\ problem was one of the problems considered by Valiant in his seminal paper on holographic algorithms~\cite{ValiantHolographic2008}. Valiant gave a polynomial time holographic algorithm for determining if a planar graph with $\Delta \leq 3$ is \fotwocolor able. Essentially, the algorithm counts the number of possible \fotwocolor ings by transforming the problem to that of counting perfect matchings. However, each \fotwocolor ing may be counted multiple times, thus we are unable to exactly count the number of \fotwocolor ings. Even so, the input graph is \fotwocolor able if and only if the sum is nonzero. Assuming that the input graph is planar, this sum can be computed in polynomial time using the FKT algorithm~\cite{Kasteleyn1961, Temperley1961, kasteleyn1963, kasteleyn1967}.

Additionally, Valiant showed that the \fotwocolor\ problem is \textbf{NP}-complete for planar graphs with $\Delta \geq 10$, and as we will explain in the next section, a recent result by Corrêa~et~al.~\cite{CorreaBrooks2009} implies the \textbf{NP}-completeness for $\Delta \geq 8$. This leaves open the hardness of the problem for graphs of maximum degree between 4 and 7.

In this paper we close the hardness gap by showing that the answer to the \fotwocolor\ problem is always yes for arbitrary graphs with $\Delta \leq 5$ and that the problem becomes \textbf{NP}-complete for planar graphs with $\Delta \geq 6$. We also observe that previous results imply that for graphs with $\Delta \leq 5$, an \fotwocolor ing can be generated efficiently using a simple greedy algorithm.

\subsection{Related Work}
It is easy to decide if a graph has a proper 2-coloring, i.e., is bipartite, since this is the case if and only if it has no cycles of odd length. However, the more general problem of deciding whether a graph has a $k$-coloring, such that the induced subgraphs all satisfy some given property $\pi$, is often significantly harder or intractable, even for $k=2$. The \fotwocolor\ problem belongs to this class of problems, where $\pi$ is the property that the induced subgraphs have a full functional orientation.

This class of graph coloring problems has been studied for arbitrary properties~\cite{Brown1987,Frick1993survey}, as well as specific properties, e.g., the induced subgraphs must be acyclic~\cite{Chartrand1968point,Chartrand1969planar} or complete~\cite{Albertson1989}. A large number of these properties involve avoiding certain induced subgraphs. This has been studied for both finite and infinite families of forbidden graphs~\cite{ Broersma2006} as well as for some fixed graph~\cite{Chartrand1968,Achlioptas1997}. 

Another and widely studied problem of this type, which is more closely related to the \fotwocolor\ problem, is that of \emph{defective coloring}~\cite{andrews1985,harary1985,cowen1986,CowenDefective1997,CorreaBrooks2009}. This is the problem of coloring the vertices with $k$ colors such that every induced subgraph has maximum degree $d$. If such a coloring exists we say that the graph is $(k,d)$-colorable. Cowen~et~al.~\cite{CowenDefective1997} showed that deciding if a planar graph of maximum degree 5 is $(2,1)$-colorable is \textbf{NP}-complete.

Valiant's proof of \textbf{NP}-completeness for \fotwocolor\ was based on the following reduction from $(2,1)$-coloring: Let $G$ be an instance for $(2,1)$-coloring. Then the graph $G'$, obtained from $G$ by duplicating every edge in $G$, is \fotwocolor able if and only if $G$ is $(2,1)$-colorable. This implies that \fotwocolor\ is \textbf{NP}-complete for planar graphs of maximum degree 10. Recently, Corrêa~et~al.~\cite{CorreaBrooks2009} improved the work of Cowen~et~al.~\cite{CowenDefective1997} by showing that $(2,1)$-coloring is \textbf{NP}-complete already for planar graphs of maximum degree 4. By Valiant's reduction this implies that \fotwocolor\ is \textbf{NP}-complete for planar graphs of maximum degree 8.

The \fotwocolor\ problem is also related to the maximum cut problem, since for $\Delta \leq 5$ a maximum cut is a $(2,2)$-coloring, and therefore also an \fotwocolor ing. This follows from the fact that the induced subgraphs are either paths or cycles, which trivially have a full functional orientation. However, unless $G$ is planar, this does not imply an efficient algorithm for finding an \fotwocolor ing, since finding a maximum cut is known to be \textbf{NP}-hard even for simple graphs of maximum degree 3~\cite{Yannakakis78}.

Lovász~\cite{LovaszDecomposition1966} showed that any graph of maximum degree $\Delta$ can be $(k,\lfloor \Delta / k \rfloor)$-colored, and Cowen~et~al.~\cite{CowenDefective1997} noted that for graphs on $v$ vertices, such a coloring can be found in $O(\Delta v)$ time using a simple greedy algorithm: Initially, let all vertices have the same color. Repeatedly, pick a vertex $u$ having more than $\lfloor \Delta / k \rfloor$ neighbours with the same color as $u$. If there is no such $u$ the coloring is a $(k,\lfloor \Delta / k \rfloor)$-coloring. Otherwise, there must exist a different color for $u$ such that at most $\lfloor \Delta / k \rfloor$ of $u$'s neighbours have this color.

Lovász's result implies that $(2,1)$-coloring is not \textbf{NP}-complete for $\Delta = 3$, and hence the simple reduction, given by Valiant, will not work to prove \textbf{NP}-completeness of \fotwocolor\ for $\Delta \geq 6$.

\subsection{Our Results}
Our main result is to settle the hardness question of the \fotwocolor\ problem for graphs of maximum degree $\Delta$. We show the following theorem
\begin{theorem} Let $G$ be a multigraph with $v$ vertices and maximum degree $\Delta$.
\begin{itemize}\label{thm:main}
\item[(i)] If $\Delta \leq 5$ then $G$ can be \fotwocolor ed in $O(v)$ time.
\item[(ii)] If $\Delta \geq 6$ the \fotwocolor\ problem is \textbf{NP}-complete, even for planar graphs.
\end{itemize}
\end{theorem}
\autoref{thm:main}(i) follows immediately from the results of Lovász~\cite{LovaszDecomposition1966} and Cowen~et al.~\cite{CowenDefective1997} for $k=2$. Considering that Valiant~\cite{ValiantHolographic2008} gave an involved decision algorithm for the case of planar graphs of maximum degree 3, it is perhaps surprising that arbitrary graphs of maximum degree 5 always have an \fotwocolor ing.

In the remaining part of the paper we prove \autoref{thm:main}(ii) by a reduction from 3-SAT. The \textbf{NP}-completeness for arbitrary graphs with $\Delta \geq 6$ is first established by a construction similar to those by Cowen~et~al.~\cite{CowenDefective1997} and Corrêa~et~al.~\cite{CorreaBrooks2009}. We extend the proof to hold for planar graphs by giving a planar crossover gadget, which we use to resolve any crossing edges.

\section{\textbf{NP}-completeness of \fotwocolor\label{sec:npc}}

In this section we establish \autoref{thm:main}(ii) by a reduction from 3-SAT in conjunctive normal form (3-CNF). To do so, we will, given an instance $\Phi$ of 3-CNF, construct a graph $G_\Phi$, which has an \fotwocolor ing if and only if $\Phi$ is satisfiable. An example of our construction is given in \autoref{fig:NPCex}. To construct such a graph we require OR-gadgets for choice, VAR-gadgets for consistency and EQ-gadgets to connect VAR-gadgets to OR-gadgets. The main challenge is to construct these gadgets such that $G_\Phi$ has maximum degree 6.

  \begin{figure}\centering
   \begin{tikzpicture}[every node/.style=defnode]

    \draw[dashed] (-0.3,-0.9) rectangle (1.3,0.3);
    \draw[dashed] (0.5,0.3) -- (0.5,-0.9);
    \draw (0,-0.9) node[draw=none,fill=none,label={below:$x_1$}] {};
    \draw (1,-0.9) node[draw=none,fill=none,label={below:$\overline{x_1}$}] {};

    \draw[dashed] (2.7,-0.9) rectangle (4.3,0.3);
    \draw[dashed] (3.5,0.3) -- (3.5,-0.9);
    \draw (3,-0.9) node[draw=none,fill=none,label={below:$x_2$}] {};
    \draw (4,-0.9) node[draw=none,fill=none,label={below:$\overline{x_2}$}] {};

    \draw[dashed] (5.7,-0.9) rectangle (7.3,0.3);
    \draw[dashed] (6.5,0.3) -- (6.5,-0.9);
    \draw (6,-0.9) node[draw=none,fill=none,label={below:$x_3$}] {};
    \draw (7,-0.9) node[draw=none,fill=none,label={below:$\overline{x_3}$}] {};

    \draw[dashed] (-0.3,1.7) rectangle (2.3,4.3);
    \draw (1,4.3) node[draw=none,fill=none,label={above:$c_1$}] {};
    \draw[dashed] (4.7,1.7) rectangle (7.3,4.3);
    \draw (6,4.3) node[draw=none,fill=none,label={above:$c_2$}] {};

    \draw (0,0) node(x11) {};
    \draw (1,0) node(x12) {};
    \draw (3,0) node(x21) {};
    \draw (4,0) node(x22) {};
    \draw (6,0) node(x31) {};
    \draw (7,0) node(x32) {};
    \draw[dedge] (x11) to[bend right=90] node{\tiny VAR} (x12);
    \draw[dedge] (x21) to[bend right=90] node{\tiny VAR} (x22);
    \draw[dedge] (x31) to[bend right=90] node{\tiny VAR} (x32);

    \draw (0,2) node(c11) {};
    \draw (1,2) node(c12) {};
    \draw (2,2) node(c13) {};
    \draw (0,4) node(c14) {};

    \draw (5,2) node(c21) {};
    \draw (6,2) node(c22) {};
    \draw (7,2) node(c23) {};
    \draw (7,4) node(c24) {};
    \draw[dedge] (c11) to node(c15){\tiny OR} (c14);
    \draw[dedge] (c13) to node(c16){\tiny OR} (c15);
    \draw[dedge] (c12) to (c16);
    \draw[dedge] (c23) to node(c25){\tiny OR} (c24);
    \draw[dedge] (c21) to node(c26){\tiny OR} (c25);
    \draw[dedge] (c22) to (c26);

    \draw[dedge] (x11) to node{\tiny EQ} (c11);
    \draw[dedge] (x12) to node{\tiny EQ} (c21);
    \draw[dedge] (x21) to node{\tiny EQ} (c12);
    \draw[dedge] (x22) to node{\tiny EQ} (c22);
    \draw[dedge] (x31) to node{\tiny EQ} (c13);
    \draw[dedge] (x32) to node{\tiny EQ} (c23);

    \draw[dedge] (c14) to node{\tiny EQ} (c24);
   \end{tikzpicture}
   \caption{The graph $G_\Phi$ generated for the 3-SAT instance \( \Phi = \left( x_1 \vee x_2 \vee x_3 \right) \wedge \left( \overline{x_1} \vee \overline{x_2} \vee \overline{x_3} \right) \), where \(x_1\),  \(x_2\) and \(x_3\) are the variables and \(c_1\) and \(c_2\) are the clauses.}
   \label{fig:NPCex}
  \end{figure}
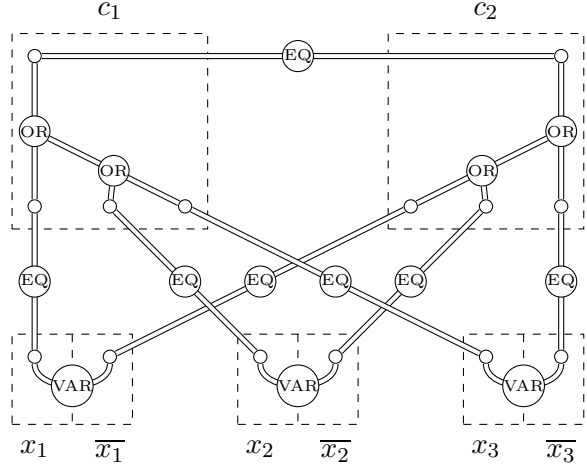

First, in \autoref{sec:npc-preliminaries} we will characterize full functional orientations of graphs and prove the existence of the VAR, OR and EQ gadgets. In \autoref{sec:npc-arbitrary} we prove that the construction of $G_\Phi$ implies \textbf{NP}-completeness of \fotwocolor\ for arbitrary graphs of maximum degree 6. Finally, in \autoref{sec:npc-planar} we provide a planar crossover gadget of maximum degree 6 to eliminate edge-crossings in $G_\Phi$, thereby proving that the \fotwocolor\ problem is \textbf{NP}-complete for planar graphs of maximum degree 6. 

\subsection{Preliminaries}\label{sec:npc-preliminaries}
The following lemma will be useful.
\begin{lemma}\label{lem:fullfo}
$G$ has a full functional orientation if and only if it consists of acyclic and unicyclic components.
\end{lemma}
\begin{proof}
It is easy to see that $G$ has a full functional orientation if it consists of acyclic and unicyclic components. Conversely, suppose that $G$ has a full functional orientation and a component containing two or more cycles. This component has strictly fewer vertices than edges and thus, by the pigeonhole principle, there is a vertex with two edges directed away from it, which is a contradiction.
\end{proof}

\begin{corollary}\label{lem:solution}
$G$ is \fotwocolor able if and only if it has a 2-coloring where every induced monochromatic component is acyclic or unicyclic.
\end{corollary}

\noindent We now show how to construct the EQ, VAR and OR gadgets. Each gadget will be a planar embedded graph $G=(V,E)$ and its unique face of infinite area is called the \emph{external face}. A subset of the vertices on the external face \( V' \subseteq V \) will be called the \emph{external vertices} of the gadget. The \emph{connection degree} of a gadget is the maximum degree of its external vertices. Gadgets are combined by identifying the external vertices. 

\begin{lemma}\label{lem:notgadget}
The \emph{NOT}-gadget of connection degree 3, shown below, ensures that \( x \) and \( y \) have different colors in any \fotwocolor ing.
\begin{center}
\begin{tikzpicture}
\draw (1,-0.5) node[anchor=south east] {\tiny NOT};
\begin{scope}[every node/.style=defnode]
\draw[dashed] (-1,-0.5) rectangle (1,0.5);
\draw (1,0) node[ionode,label={right:$y$}](y) {};
\draw (-1,0) node[ionode,label={left:$x$}](x) {};
\draw (x) to [bend left=20] (y);
\draw (x) to (y);
\draw (x) to [bend right=20] (y);
\end{scope}
\end{tikzpicture}
\end{center}
\end{lemma}
\begin{proof}
This follows trivially from \autoref{lem:solution}.
\end{proof}

\begin{lemma}\label{lem:eqgadget}
The \emph{EQ}-gadget of connection degree \( 2 \), shown below, ensures that \( x \) and \( y \) have the same color in any \fotwocolor ing.
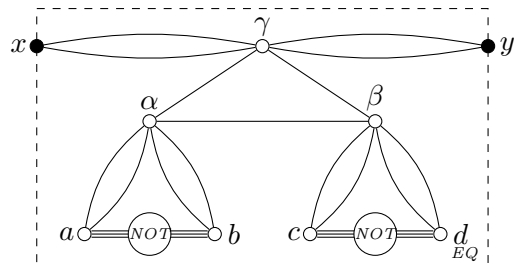
\begin{center}
\begin{tikzpicture}
\draw (3,-2) node[anchor=south east] {\tiny EQ};
\draw[dashed] (-3,-2) rectangle (3,1.5);
\begin{scope}[every node/.style=defnode]
\draw (3,1) node[ionode,label={right:$y$}](y) {};
\draw (-3,1) node[ionode,label={left:$x$}](x) {};
\draw (0,1) node[label={above:$\gamma$}](7) {};

\begin{scope}[shift={(-1.5,-1)}]
\draw (90:1) node[label={above:$\alpha$}](1) {};
\draw (210:1) node[label={left:$a$}](2) {};
\draw (330:1) node[label={right:$b$}](3) {};
\draw[tedge] (2) to (3);
\draw (2) to node{\tiny NOT} (3);
\draw (3) to [bend left=20] (1);
\draw (3) to [bend right=20] (1);
\draw (1) to [bend left=20] (2);
\draw (1) to [bend right=20] (2);
\end{scope}
\begin{scope}[shift={(1.5,-1)}]
\draw (90:1) node[label={above:$\beta$}](4) {};
\draw (210:1) node[label={left:$c$}](5) {};
\draw (330:1) node[label={right:$d$}](6) {};
\draw[tedge] (5) to (6);
\draw (5) to node{\tiny NOT} (6);
\draw (6) to [bend left=20] (4);
\draw (6) to [bend right=20] (4);
\draw (4) to [bend left=20] (5);
\draw (4) to [bend right=20] (5);
\end{scope}

\draw [bend left=10] (x) to (7);
\draw [bend left=10] (7) to (x);
\draw [bend left=10] (y) to (7);
\draw [bend left=10] (7) to (y);
\draw (7) to (1) to (4) to (7);
\end{scope}

\end{tikzpicture}
\end{center}
\end{lemma}
\begin{proof}
Assume there exists a coloring where \( x \) and \( y \) have different colors. Assume without loss of generality that \( x \) is colored \( 0 \) and \( y \) is colored \( 1 \). Also assume, again without loss of generality, that \( a \), \( c \) and \( \gamma \) are colored \( 0 \) and \( b \) and \( d \) are colored \( 1 \). All four possible colorings of \( \alpha \) and \( \beta \) induce a monochromatic component with two cycles, hence violating \autoref{lem:solution}. Conversely, assume that \( x \) and \( y \) have the same color, without loss of generality assume that color to be \( 0 \). Then \( \gamma \), \( \alpha \), \( a \) and \( c \) may be colored \( 1 \) and \( \beta \), \( b \) and \( d \) may be colored \( 0 \). This is a valid coloring.
\end{proof}

An essential property of the EQ-gadget is that in any \fotwocolor ing of the gadget it holds that $\gamma$ has the opposite color of that of $x$ and $y$. This allows us to connect arbitrary gadgets through intermediate EQ-gadgets, since the orientation of the external vertices are never used internally in the EQ-gadget.

\begin{lemma}\label{lem:negadget}
The NE-gadget of connection degree \( 2 \), shown below, ensures that \( x \) and \( y \) have different colors in any \fotwocolor ing.
\begin{center}
\begin{tikzpicture}
\draw[dashed] (-3,-0.5) rectangle (3,0.5);
\draw (3,-0.5) node[anchor=south east] {\tiny NE};
\begin{scope}[every node/.style=defnode]
\draw (3,0) node[ionode,label={right:$y$}](y) {};
\draw (1,0) node(1) {};
\draw (-1,0) node(2) {};
\draw (-3,0) node[ionode,label={left:$x$}](x) {};
\draw[dedge] (1) to node{\tiny EQ} (y);
\draw[dedge] (2) to node{\tiny EQ} (x);
\draw[tedge] (1) to (2);
\draw (1) to node{\tiny NOT} (2);
\end{scope}
\end{tikzpicture}
\end{center}
\end{lemma}

\begin{proof}
This follows trivially from \autoref{lem:eqgadget} and \ref{lem:negadget}.
\end{proof}

\begin{lemma}
  The planar OR-gadget of connection degree \( 2 \), shown below, ensures that in any \fotwocolor ing the color of \( z \) is the same as the one of \( x \) or the one of \( y \).
 
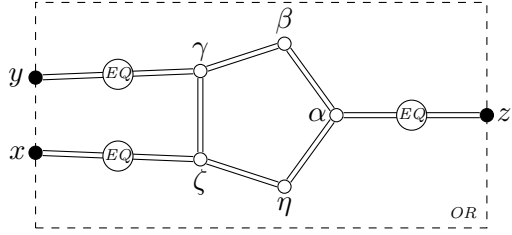
\label{lem:orgadget}
\begin{center}
\begin{tikzpicture}
\draw[dashed] (-3,-1.5) rectangle (3,1.5);
\draw (3,-1.5) node[anchor=south east] {\tiny OR};
\begin{scope}[every node/.style=defnode]
\draw (-3,-0.5) node[ionode,label={left:$x$}](x) {};
\draw (-3,0.5) node[ionode,label={left:$y$}](y) {};
\draw (0:1) node[label={left:$\alpha$}](4) {};
\draw (72:1) node[label={above:$\beta$}](5) {};
\draw (144:1) node[label={above:$\gamma$}](1) {};
\draw (216:1) node[label={below:$\zeta$}](2) {};
\draw (288:1) node[label={below:$\eta$}](3) {};
\draw (3,0) node[ionode,label={right:$z$}](z) {};
\draw[dedge] (1) to node{\tiny EQ} (y);
\draw[dedge] (2) to node{\tiny EQ} (x);
\draw[dedge] (4) to node{\tiny EQ} (z);
\draw[dedge] (1) to (2) to (3) to (4) to (5) to (1);
\end{scope}
\end{tikzpicture}
\end{center}
\end{lemma}

\begin{proof}
Assume for the sake of a contradiction that the external vertices \( x \) and \( y \) have the same color and \( z \) the other color. Then \( \gamma\), \(\zeta \) and \( \alpha \) have the same color as \( x\), \( y \) and \( z \), respectively. All four possible colorings of \( \beta \) and \( \eta \) induce a monochromatic component with two cycles, hence violating \autoref{lem:solution}. It is easy to verify that any other coloring of the external vertices, \(x, y \) and \(z \) is consistent with an \fotwocolor ing.
\end{proof}

\begin{lemma}\label{lem:vargadget}
The VAR$_{n,m}$-gadget of connection degree \( 4 \), shown below, ensures that in any \fotwocolor ing \( x_1,x_2,\ldots, x_n \) have the same color and \( \overline{x}_1, \overline{x}_2, \ldots, \overline{x}_m \) have the opposite color.
\begin{center}
\begin{tikzpicture}
\draw[dashed] (-5,-0.5) rectangle (5,0.5);
\draw (5,-0.5) node[anchor=south east] {\tiny VAR$_{n,m}$};
\begin{scope}[every node/.style=defnode]
\draw (3,-0.5) node[ionode,label={[label distance=-0.14cm]below:$\overline{x}_{m-1}$}](y) {};
\draw (2,0) node(EQ2) {\tiny EQ};
\draw (1,-0.5) node[ionode,label={below:$\overline{x}_m$}](1) {};

\draw (-1,0.5) node[ionode,label={above:$x_{n}$}](2) {};
\draw (-2,0) node(EQ1) {\tiny EQ};
\draw (4,0) node[draw=none](3) {};
\draw (-4,0) node[draw=none](4) {};
\draw (-3,0.5) node[ionode,label={[label distance=-0.17cm]above:$x_{n-1}$}](x) {};

\draw[dedge] (1) to (EQ2) to (y);
\draw[dedge] (2) to (EQ1) to (x);
\draw[dedge] (1) to node{\tiny NE} (2);
\draw[dashed,dedge] (y) to (3);
\draw[dashed,dedge] (x) to (4); 
\end{scope}
\end{tikzpicture}
\end{center}
\end{lemma}
\begin{proof}
This follows trivially from \autoref{lem:eqgadget} and \ref{lem:negadget}.
\end{proof}

\subsection{NP-completeness for Arbitrary Graphs}\label{sec:npc-arbitrary}

\begin{lemma}\label{thm:fo2color-delta6-anpc}
The \fotwocolor\ problem is \textbf{NP}-complete for graphs with \( \Delta \geq 6 \).
\end{lemma}

\begin{proof}
 We prove the \textbf{NP}-completeness of \fotwocolor\ by a reduction from 3-SAT in conjunctive normal form (3-CNF). First note that the problem is in \textbf{NP}, as the validity of an \fotwocolor ing can be verified in polynomial time.

Given an arbitrary instance \( \Phi \) of  3-CNF, we will construct a graph \( G_\Phi \) in polynomial time, such that \( G_\phi \) is \fotwocolor able if and only if \( \Phi \) is satisfiable. To construct the graph \( G_\Phi \), we use the EQ-gadget of \autoref{lem:eqgadget}, the OR-gadget of \autoref{lem:orgadget} and the VAR$_{n,m}$-gadget of \autoref{lem:vargadget}. The process is as follows: For every variable \( X \) instantiate a VAR$_{n,m}$-gadget with \( n \) being the number of clauses in which \( X \) occurs unnegated and \( m \) the corresponding negated occurrences. Then for every clause two OR-gadgets are instantiated and the output of one is connected to an input of the other, thus forming a three input OR-gadget. The remaining OR-gadget outputs are then connected together using EQ-gadgets, while the VAR-gadgets are connected with the OR-gadgets, using EQ-gadgets, such that the OR-gadget for clause \( C \) is connected with the unegated (negated) side of the VAR\(_{n,m}\)-gadget for variable \( X \) if and only if \( X \) occurs unnegated (negated) in \( C \). An example construction may be seen in \autoref{fig:NPCex}.

  Consider the case where there exists an \fotwocolor ing of the graph. In this case the final outputs of the OR-gadgets will be identically colored. Let this color correspond to true. By \autoref{lem:orgadget} at least one of the three inputs has the same color, i.e., at least one literal in every clause is true. Also \autoref{lem:vargadget} ensures that the variables are assigned consistent values.
  Conversely, consider a satisfying assignment \( \mu \) to \( \Phi \). From \( \mu \), it is possible to create an \fotwocolor ing in the following fashion: For every variable gadget, assign color 0 (1) to the \(n\) side and 1 (0) to the \(m\) side of the gadget if the variable is true (false) in \( \mu \). These are then propagated by the equality gadgets and, as the assignment was satisfying, every OR-gadget has the color 0 on at least one input, therefore the output may also be colored 0. This coloring is therefore valid.

  Therefore each \fotwocolor ing of the graph \( G_\Phi \) corresponds to a satisfying assignment of \( \Phi \). Additionally the existence of a satisfying assignment to \( \Phi \) implies the existence of an \fotwocolor ing of \( G_\Phi \). Thus showing the \textbf{NP}-completeness of \fotwocolor\ for arbitrary graphs with \( \Delta \geq 6 \).
\end{proof}

\subsection{NP-completeness for Planar Graphs}\label{sec:npc-planar}

We extend the proof of \autoref{thm:fo2color-delta6-anpc} to planar graphs by the elimination of crossing edges using a crossover gadget. If such a gadget, of maximum degree \( 6 \), exists, then \autoref{thm:main}(ii) follows.

\begin{lemma}\label{lem:xover}
The XO-gadget of connection degree \( 4 \), shown below, ensures that \( x \) and \( x' \) as well as \( y \) and \( y' \) have the same colors in any \fotwocolor ing.
\begin{center}
\begin{tikzpicture}[scale=1.4]
\draw (3.5,-3.5) node[anchor=south east] {\tiny XO};
\draw[dashed] (-3.5,-3.5) rectangle (3.5,3.5);
\begin{scope}[every node/.style=defnode]
\draw (0.5,0.5) node(1) {};
\draw (-0.5,0.5) node(2) {};
\draw (-0.5,-0.5) node(3) {};
\draw (0.5,-0.5) node(4) {};
\draw (2.5,1) node(5) {};
\draw (1,2.5) node(6) {};
\draw (-1,2.5) node(7) {};
\draw (-2.5,1) node(8) {};
\draw (-2.5,-1) node(9) {};
\draw (-1,-2.5) node(10) {};
\draw (1,-2.5) node(11) {};
\draw (2.5,-1) node(12) {};
\draw (1.5,0) node(13) {};
\draw (-1.5,0) node(14) {};
\draw (2.5,2.5) node(15) {};
\draw (-2.5,3.5) node[ionode,label={above:$x$}](16) {};
\draw (-2.5,-2.5) node(17) {};
\draw (2.5,-3.5) node[ionode,label={below:$x'$}](18) {};
\draw (2.5,3.5) node[ionode,label={above:$y$}](19) {};
\draw (-2.5,-3.5) node[ionode,label={below:$y'$}](20) {};
\draw[dedge] (1) to node{\tiny EQ} (2) to node{\tiny EQ} (3) to node{\tiny EQ} (4) to node{\tiny EQ} (1);
\draw (5) to (6) to (7) to (8) to (9) to (10) to (11) to (12) to (5);
\draw[dedge] (1) to node{\tiny NE} (6);
\draw[dedge] (2) to node{\tiny EQ} (14) to node{\tiny EQ} (8);
\draw[dedge] (3) to node{\tiny NE} (10);
\draw[dedge] (4) to node{\tiny EQ} (13) to node{\tiny EQ} (12);
\draw[bend left=15] (13) to (5) to (13);
\draw[bend left=15] (14) to (9) to (14);

\draw[dedge] (7) to node{\tiny EQ} (16);
\draw[bend left=15] (8) to (16) to (8);
\draw[bend left=7] (7) to (19) to (7);
\draw[dedge] (6) to node{\tiny EQ} (15) to node{\tiny EQ} (19);
\draw[bend left=15] (5) to (15) to (5);

\draw[dedge] (11) to node{\tiny EQ} (18);
\draw[bend left=15] (12) to (18) to (12);
\draw[bend left=7] (11) to (20) to (11);
\draw[dedge] (10) to node{\tiny EQ} (17) to node{\tiny EQ} (20);
\draw[bend left=15] (9) to (17) to (9);
\end{scope}
\end{tikzpicture}
\end{center}
\end{lemma}
\begin{proof}
By careful inspection along with \autoref{lem:eqgadget} and \ref{lem:negadget}, it can be seen that there are only 4 \fotwocolor ings, all of which satisfy the lemma.
\end{proof}

\begin{proof}[Proof of \autoref{thm:main}(ii)]
Given a 3-CNF instance \( \Phi \), we apply the reduction of \autoref{thm:fo2color-delta6-anpc}. The resulting graph \( G_\Phi \) then contains at most \( 6 \left| C \right| \) edges which connect OR-gadgets to VAR\(_{n,m}\)-gadgets. Each such edge can at most cross every other such edge once, except the corresponding parallel edge. These are the only crossings in \( G_\Phi \), thus there are at most \( 18 \left| C \right|^2 - 6 \left| C \right| \) crossings, each of which is replaced by an XO-gadget resulting in \( G'_\Phi \). The graph \( G'_\Phi \) is planar and preserves the solutions of \( G_\Phi \) by \autoref{lem:xover}. Consequently, \fotwocolor\ is \textbf{NP}-complete for planar graphs with \( \Delta \geq 6 \).
\end{proof}

\section*{Acknowledgements}
We are thankful to Philip Bille, Thore Husfeldt, Konstantin Kutzkov, Rasmus Pagh, Carsten Thomassen and the reviewers for their valuable comments on an earlier draft of this paper.

\bibliographystyle{abbrv}
\bibliography{paper}

\end{document}